\documentclass[a4paper]{article}

\usepackage[margin=35mm]{geometry}
\usepackage{lmodern}
\usepackage[T1]{fontenc}
\usepackage[utf8]{inputenc}
\usepackage[hidelinks]{hyperref}
\usepackage{amsmath, amssymb, amsthm, thm-restate, enumerate, tcolorbox, enumitem}
\usepackage{cleveref}

\setlist[enumerate]{topsep=2pt, partopsep=0pt, itemsep=2pt, parsep=0pt}

\theoremstyle{plain}
\newtheorem{theorem}{Theorem}
\newtheorem{corollary}[theorem]{Corollary}
\newtheorem{lemma}[theorem]{Lemma}

\newenvironment*{claim*}[1][]{%
  \par\vspace{0.5em}\noindent \textit{Claim\ifx\empty#1\else\textnormal{ (#1)}\fi.} \ignorespaces
}{%
  \par\vspace{0.5em}
}
\newenvironment{claimproof}[1][Proof of claim]{
  \par\noindent \textit{#1.} \ignorespaces
}{
  \hfill {\large$\triangleleft$}\par\vspace{0.5em}
}

\theoremstyle{definition}
\newtheorem{definition2}[theorem]{Definition}

\newcommand{\R}{\mathbb{R}}
\newcommand{\N}{\mathbb{N}}
\DeclareMathOperator{\OPT}{\mathtt{OPT}}
\DeclareMathOperator{\MST}{\mathtt{MST}}
\DeclareMathOperator{\E}{\mathbb{E}}
\newcommand{\A}{\mathcal{A}}
\newcommand{\X}{\mathcal{X}}
\newcommand{\T}{\mathcal{T}}
\newcommand{\Low}{\mathcal{L}}
\newcommand{\High}{\mathcal{H}}

\title{Online Metric TSP}

\author{Christian Bertram\thanks{University of Copenhagen, Denmark, \texttt{chbe@di.ku.dk}, \href{https://orcid.org/0009-0009-7940-1002}{ORCID 0009-0009-7940-1002}. The author is part of BARC, Basic Algorithms Research Copenhagen, supported by the VILLUM Foundation grant 54451.}}

\begin{document}

\maketitle

\begin{abstract}
  In the \emph{online metric traveling salesperson problem}, $n$ points of a metric space arrive one by one and have to be placed (immediately and irrevocably) into empty cells of a size-$n$ array.
  The goal is to minimize the sum of distances between consecutive points in the array.
  This problem was introduced by Abrahamsen, Bercea, Beretta, Klausen, and Kozma [ESA'24] as a generalization of the \emph{online sorting problem}, which was introduced by Aamand, Abrahamsen, Beretta, and Kleist [SODA'23] as a tool in their study of online geometric packing problems.

  Online metric TSP has been studied for a range of fixed metric spaces.
  For 1-dimensional Euclidean space, the problem is equivalent to online sorting, where an optimal competitive ratio of $\Theta(\sqrt n)$ is known.
  For $d$-dimensional Euclidean space, the best-known upper bound is $O(2^{d} \sqrt{dn\log n})$, leaving a gap to the $\Omega(\sqrt n)$ lower bound.
  Finally, for the uniform metric, where all distances are 0 or 1, the optimal competitive ratio is known to be $\Theta(\log n)$.

  We study the problem for a general metric space, presenting an algorithm with competitive ratio $O(\sqrt n)$.
  In particular, we close the gap for $d$-dimensional Euclidean space, completely removing the dependence on dimension.
  One might hope to simultaneously guarantee competitive ratio $O(\sqrt n)$ in general and $O(\log n)$ for the uniform metric, but we show that this is impossible.
\end{abstract}

\section{Introduction}%
\label{sec:introduction}

\subsection{Problem definition}

The \emph{online metric traveling salesperson problem} (online metric TSP), recently introduced by Abrahamsen, Bercea, Beretta, Klausen, and Kozma~\cite{online-tsp}, is as follows.
Given a sequence $x_{1}, \ldots, x_{n}$ of points arriving one by one (with repetitions allowed) from a metric space $(M, d)$, assign them bijectively to array cells $A[1], \ldots, A[n]$.
The goal is to minimize $\sum_{i = 1}^{n-1}d(A[i], A[i+1])$, which represents the length of the walk $A[1], \ldots, A[n]$.
The problem is \emph{online} in the sense that after receiving $x_{i}$, we must immediately and irrevocably set $A[j] = x_{i}$ for some previously unused array index $j$, without knowledge of $x_{k}$ for $k > i$.

We can think of this as a metric traveling salesperson problem, where $n$ cities are sequentially revealed, one by one, and must be placed on a unique date in the salesperson's $n$-day calendar.
The cost to be minimized is the length of the final $n$-day trip.

In~\cite{online-tsp}, the metric space is fixed as 1-dimensional Euclidean space, $d$-dimensional Euclidean space, or a space with uniform/discrete metric (where all distances are 0 or 1).
In this paper, we study the problem for a general metric space, allowing the algorithm to query the distance $d(x_{i}, x_{j})$ between $x_{i}$ and $x_{j}$, if it has received $x_{i}$ and $x_{j}$.
Alternatively, the $i$th input can be assumed to be the vector $(d(x_{i}, x_{1}), d(x_{i}, x_{2}), \ldots, d(x_{i}, x_{i-1}))$.

The term ``online TSP'' has also been used for a different, older problem where a salesperson moves through the metric space at unit speed as cities are revealed~\cite{oltsp}.
In contrast, our problem consists of constructing a fixed travel plan through irrevocable online decisions.

\subsection{Prior work}

In~\cite{online-packing}, the \emph{online sorting problem} is introduced as part of their study of various geometric translational packing problems.
The online sorting problem is equivalent to online (metric) TSP in the Euclidean unit interval $[0, 1]$.
They further essentially assume that the points $0$ and $1$ always show up in the input, which simplifies the analysis, as the optimal offline cost becomes 1.
They present a deterministic algorithm with competitive ratio $O(\sqrt n)$ for online sorting (where \emph{competitive ratio} is the worst-case ratio between the algorithm's cost and the optimal offline cost), and show an $\Omega(\sqrt n)$ lower bound for deterministic algorithms.

In~\cite{online-tsp}, the online sorting algorithm from~\cite{online-packing} is generalized to points from the real line $\R$, not necessarily including 0 and 1, by employing a careful \emph{doubling technique}.
They maintain the $O(\sqrt n)$ competitive ratio, and generalize the $\Omega(\sqrt n)$ lower bound to randomized algorithms.

In~\cite{online-tsp}, they further generalize the problem to online metric TSP, and consider this problem for a few fixed metric spaces.
For $d$-dimensional Euclidean space, they present an $O(2^{d}\sqrt{dn \log n})$ competitive ratio algorithm, leaving a gap to their $\Omega(\sqrt n)$ lower bound.
For use as a subroutine in this algorithm, they consider online TSP in the uniform metric, i.e.\ the metric where all distances are 0 or 1.
They argue that this problem has interest in its own right, corresponding to minimizing the number of task switches in scheduling, and present a tight $\Theta(\log n)$ bound on the competitive ratio.
Finally, they ask, as an open question, whether $O(\sqrt n)$ is the optimal competitive ratio for arbitrary metrics.

\subsection{Our results}

Our main result is an optimal algorithm for online metric TSP in a general metric space, settling a question posed in~\cite{online-tsp}.

\begin{theorem}\label{sqrt-n-alg}
  There exists a deterministic algorithm for online metric TSP with competitive ratio $O(\sqrt n)$.
\end{theorem}

This is optimal by the $\Omega(\sqrt n)$ lower bound for the Euclidean unit interval~\cite{online-packing}, which holds even for randomized algorithms~\cite{online-tsp}.
As a direct corollary, we get an optimal algorithm for online TSP in $d$-dimensional Euclidean space, improving upon the best-known $O(2^{d}\sqrt{dn \log n})$ competitive ratio algorithm~\cite{online-tsp}, completely removing the dependence on dimension.

\begin{corollary}\label{sqrt-n-alg-euclidean}
  There exists a deterministic algorithm for online TSP in $d$-dimensional Euclidean space with competitive ratio $O(\sqrt n)$.
\end{corollary}

Note, though, that the $O(\sqrt n)$ bound of \Cref{sqrt-n-alg} does not match the known $\Theta(\log n)$ bound for online TSP with uniform metric~\cite{online-tsp}.
One might hope for an even stronger algorithm, obtaining an asymptotically optimal competitive ratio for every fixed metric space.
We show that no such algorithm exists, hinting that our algorithm is the best one can hope for.

\begin{restatable}{theorem}{impossibilityResult}\label{no-log-n-and-sqrt-n}
  No randomized algorithm for online metric TSP obtains both $O(\log n)$ expected competitive ratio for the uniform metric and $O(\sqrt n)$ expected competitive ratio in general.
\end{restatable}

\subsection{Structure of the paper}

In~\Cref{sec:preliminaries}, we recall relevant definitions and fix notation.
In~\Cref{sec:our-algorithm}, we present our algorithm, proving \Cref{sqrt-n-alg,sqrt-n-alg-euclidean}.
In~\Cref{sec:optimality}, we confirm the worst-case optimality of our algorithm via known lower bounds, and then we show that no algorithm is simultaneously optimal for every fixed metric, proving \Cref{no-log-n-and-sqrt-n}.

\section{Preliminaries}%
\label{sec:preliminaries}

We will often let $(M, d)$ be a metric space, where $M$ is the set of points, and $d \colon M \times M \to \R_{\geq 0}$ is the metric.
Our main examples are $d$-dimensional Euclidean space and the uniform metric.

The \emph{uniform metric} is defined by $d(x, y) = 1$ for all $x \neq y$.
This is also known as the \emph{discrete metric}, but we adopt the former term for consistency with~\cite{online-tsp}.
Any set of points can form a metric space under the uniform metric.

Though our algorithms run in polynomial time, we do not focus on their exact running times.
Instead, we evaluate performance via \emph{competitive analysis}, as is standard in the study of online algorithms.
Let $\OPT(X)$ denote the optimal offline cost for an input sequence $X \in M^{n}$.
An online algorithm is said to have \emph{competitive ratio} $C(n)$, for a function $C \colon \N \to \R$, if its cost is at most $C(n)\OPT(X)$ for all $n \in \N$ and all $X \in M^{n}$.

In our setting, the cost is $\sum_{i=1}^{n-1}d(A[i], A[i+1])$, and $\OPT(X)$ is the minimum value of this sum, minimized over all bijections from the input points in $X$ to array cells $A[1],\ldots,A[n]$.
By the triangle inequality, $\OPT(X)$ is equal to the length of a shortest walk visiting all points in $X$.
We will sometimes write $\OPT(X)$ for a subset $X \subseteq M$, as the optimal offline cost is invariant under permutations and repetitions of the input points.

\section{Our algorithm}%
\label{sec:our-algorithm}

Our algorithm is based on the algorithms for online sorting presented in~\cite{online-packing} and~\cite{online-tsp}.
Let us first give a high-level explanation of their approach.
Denote by $X'$ the set of currently received points.
They consider an interval $J$ containing $X'$, and partition $J$ into $\sqrt n$ subintervals.
Specifically, they assume without loss of generality that $x_{1} = 0$, and let $J$ be the smallest interval of the form $[-2^{k},2^{k}]$ containing $X'$, increasing $k$ as necessary.
Note that $\OPT(X') \geq 2^{k-1}$, so the subintervals have size $O(\OPT(X')/\sqrt{n})$.
Similarly, they partition the array $A$ into $2\sqrt n$ subarrays, called $\emph{blocks}$.
The algorithm generally places the input points into the array, such that no block contains points from different subintervals.
This ensures, that they only pay a small $O(\OPT(X')/\sqrt{n})$ cost between points inside each block, and pay a large $O(\OPT(X'))$ cost only between blocks.
The former cost is paid $O(n)$ times, and the latter $O(\sqrt n)$ times, totaling a cost of $O(\sqrt n\OPT(X'))$.
It is always possible to place the first half of the points in this manner, after which they consider the remaining empty array cells as one contiguous array, and recursively fill this simulated array with the remaining input points.
This leads to a total cost of $O(\sqrt n\OPT(X))$ as the cost per recursion falls geometrically.
Each time $k$ is increased, we incur some additional cost of changing the intervals, but this only happens when the bound on $\OPT(X')$ doubles, so again it all sums to $O(\sqrt n\OPT(X))$.

Let us now give a high-level explanation of how we will extend this algorithm from the Euclidean line to a general metric space.
There are two main problems to solve: the setting of $J=[-2^{k}, 2^{k}]$ and bound of $\OPT(X') \geq 2^{k-1}$ no longer make sense.
One might naively try letting $J$ be a ball of radius $2^{k}$, but this cannot generally be covered by $\sqrt n$ balls of radius $O(2^{k}/\sqrt n)$.
To see this, think of $d$-dimensional Euclidean space, where $2^{d}$ unit cubes are required to cover a cube of sidelength 2.
This is related to why there is a term of size $2^{d}$ in the best-known bound for $d$-dimensional Euclidean space~\cite{online-tsp}.
We instead present an online algorithm maintaining a set of at most $\sqrt n$ balls of radius $r=\Theta(\OPT(X')/\sqrt{n})$ covering all received input points.
Only when $\OPT(X')$ doubles, we allow changing the radius $r$ and resetting the set.
Of course, we don't know the value of $\OPT(X')$, but it turns out that we can successfully use the minimum spanning tree weight as a (polynomial-time computable) proxy.
This proxy will be essential when showing that our covering only needs $\sqrt n$ balls.

We begin the formal presentation of our algorithm with the above mentioned proxy.

\begin{definition2}
  Let $(M, d)$ be a metric space and $X \subseteq M$ a finite subset of points.
  When the metric $d$ is clear from context, we denote by $\MST(X)$ the total weight of a minimum spanning tree of the complete graph on $X$, where edge weights are given by $d$.
\end{definition2}

Computing exact metric TSP length is NP-hard, as follows by a straightforward reduction from the undirected Hamiltonian cycle problem, one of Karp's 21 NP-complete problems~\cite{karp}.
Luckily, there is a simple 2-approximation based on minimum spanning trees~\cite{2-approx}.
This well-known argument yields the following \lcnamecref{2-approx}, which tells us that $\MST(X)$ is a good proxy of $\OPT(X)$.
We provide a proof in~\Cref{app:proxy-details}.

\begin{restatable}{lemma}{proxy}\label{2-approx}
  Let $(M, d)$ be a metric space and $X \subseteq M$ a finite subset of points.
  Then $\MST(X) \leq \OPT(X) \leq 2\MST(X)$.
\end{restatable}

As sketched, our algorithm will cover the input points by a dynamic set of balls.
Formally, we will work with a set of \emph{centers} forming a \emph{net} in the following standard manner.

\begin{definition2}
  Let $(M, d)$ be a metric space and $X \subseteq M$ a set of points.
  Let $r \geq 0$ be a real number.
  A subset $C \subseteq X$ is an \emph{$r$-net} of $X$ if
  \begin{itemize}[topsep=1pt,itemsep=1pt,parsep=1pt]
    \item for every $x \in X$, there exists a $c \in C$ with $d(c, x) \leq r$, and
    \item for every $c, c' \in C$ with $c \neq c'$, we have $d(c, c') > r$.
  \end{itemize}
  If so, we sometimes refer to each point in $C$ as a \emph{center}, and to $r$ as the \emph{radius} of the net.
\end{definition2}

Nets will be useful to us, because they are small in terms of our proxy, in the sense of the following \lcnamecref{net-size}.

\begin{lemma}\label{net-size}
  Let $(M, d)$ be a metric space, in which $C$ is an $r$-net of a set $X \subseteq M$.
  Then $(|C|-1)r \leq 2\MST(X)$.
\end{lemma}
\begin{proof}
  By the triangle inequality, we have $\OPT(X) \geq \OPT(C)$.
  Combining this with~\Cref{2-approx}, we get $2\MST(X) \geq \OPT(X) \geq \OPT(C) \geq \MST(C)$.
  Every edge in the complete graph on $C$ has length at least $r$, and a tree on $|C|$ points has $|C|-1$ edges, so $\MST(C) \geq (|C|-1)r$.
\end{proof}

When receiving a new point $x$, we will generally update our $r$-net by the following simple subroutine.

\begin{tcolorbox}[beforeafter skip=10pt]
  \textbf{Increase-Net($C, r, x$)}
  \begin{enumerate}
    \item If no $c \in C$ has $d(x, c) \leq r$ then add $x$ to $C$.
  \end{enumerate}
\end{tcolorbox}

This subroutine allows us to maintain an $r$-net, for a fixed radius $r$.
We state this as the following~\lcnamecref{increase-net}, which follows directly from the definition.

\begin{lemma}\label{increase-net}
  Let $(M, d)$ be a metric space, in which $C$ is an $r$-net of a set $X \subseteq M$, and let $x \in X$.
  Then running \textnormal{Increase-Net$(C, r, x)$} modifies $C$ only by insertion, such that $C$ becomes an $r$-net of $X\cup\{x\}$.
\end{lemma}

To facilitate analysis, let us introduce two pieces of notation.
For a partially filled array $A$, we write $c(A) = \sum_{i \in I}d(A[i], A[i+1])$ for the \emph{cost} of $A$, where $I$ is the set of indices $i$ for which both $A[i]$ and $A[i+1]$ are non-empty.
Note that this matches the usual cost, when $A$ is full.
Secondly, we write $G(A)$ for the number of \emph{gaps} in $A$, i.e.\ the number of non-trivial maximal contiguous empty subarrays of $A$.

We now present ``half'' of our algorithm, namely an algorithm handling the first $\lceil n/2 \rceil$ input points.
First, we give a quick overview.
The algorithm will partition the array $A$ into $2\lfloor\sqrt n\rfloor$ blocks, and maintain an $r$-net $C$ with at most $\lfloor\sqrt n\rfloor$ centers.
If we get too many centers, we will reset the net with a larger radius.
Each block may be \emph{assigned} to a center, such that every center is assigned at most one block.
Initially, every block is unassigned.
Generally, if a newly received point lies within the ball of a center, it will be placed inside the block assigned to that center.
The exact algorithm is given below as Fill-Most-Blocks.

\begin{lemma}\label{half-alg}
  Let $(M, d)$ be a metric space and $X$ be an online stream of $\lceil n/2 \rceil$ points in $M$.
  Let $A$ be an empty array of length $n$.
  The deterministic algorithm \textnormal{Fill-Most-Blocks$(n, A, X)$} irrevocably places each point from $X$ in an empty cell of $A$, such that when all points have been placed, we have $G(A) \leq 2\sqrt n$ and $c(A) \leq 11\sqrt{n}\OPT(X)$.
\end{lemma}

\begin{tcolorbox}[beforeafter skip=10pt]
  \textbf{Fill-Most-Blocks($n, A, X$)}
  \begin{enumerate}
    \item Let $N_{1} = \lfloor \sqrt n \rfloor$ and $N_{2} = 2N_{1}$.
    \item Partition $A$ into $N_{2}$ subarrays, called \emph{blocks}, of length at least $\lfloor n/N_{2} \rfloor$.
    \item\label{alg-step:init-net} Initialize $r = 0$ and $C = \emptyset$.
    \item For each point $x$ in the stream $X$:
          \begin{enumerate}
            \item\label{alg-step:increase-net} Increase-Net$(C, r, x)$.
                  \item\label{alg-step:if-large-net} If $|C| > N_1$:
                  \begin{enumerate}
                    \item\label{alg-step:unassign-all} Unassign all blocks.
                    \item\label{alg-step:set-radius} Set $r = 4\MST(X')/N_{1}$, where $X'$ is the set of known input points.
                    \item\label{alg-step:empty-net} Set $C = \{x\}$.
                  \end{enumerate}
            \item Let $c \in C$ be a center with $d(x, c) \leq r$.
            \item If a full block $B$ is assigned to $c$:
                  \begin{enumerate}
                    \item\label{alg-step:unassign-block} Unassign $B$ from $c$.
                  \end{enumerate}
            \item If no block is assigned to $c$:
                  \begin{enumerate}
                    \item\label{alg-step:assign-block} Assign an unassigned non-full block to $c$.
                  \end{enumerate}
            \item Let $B$ be the block assigned to $c$.
                  \item\label{alg-step:place-point} Place $x$ in the left-most empty cell of $B$.
          \end{enumerate}
  \end{enumerate}
\end{tcolorbox}

\begin{proof}
  The following two claims show that the algorithm is well-defined.

  \begin{claim*}
    After step~\labelcref{alg-step:if-large-net}, $C$ is an $r$-net of a set containing $x$.
  \end{claim*}
  \begin{claimproof}
    Clearly, the claim holds if step~\labelcref{alg-step:empty-net} just ran, so let us show, that it also holds before and between these resets of the net.
    This follows by induction, where the base case is step~\labelcref{alg-step:init-net} producing an $r$-net, and the induction step is step~\labelcref{alg-step:increase-net} with~\Cref{increase-net}.
  \end{claimproof}

  \begin{claim*}
    In step~\labelcref{alg-step:assign-block}, there is always an unassigned non-full block to assign.
  \end{claim*}
  \begin{claimproof}
    Assume for the sake of contradiction, that no unassigned non-full block exists.
    Then every block is either assigned or full.
    There are at most $|C| \leq N_{1}$ assigned blocks, so at least $N_{2}-|C| \geq N_{2}-N_{1}=N_{1}$ blocks are full.
    Note that every assigned block contains at least one point, as right after assigning a block in step~\labelcref{alg-step:assign-block}, it is given a point in step~\labelcref{alg-step:place-point}.
    So at least $N_{1}$ blocks contain at least $\lfloor n/N_{2} \rfloor$ points, and the remaining blocks contain at least one point.
    This means that the total number of points in $A$ is at least $N_{1} \lfloor n/N_{2} \rfloor + N_{1} \geq N_{1}n/N_{2} = n/2$.
    This contradicts the stream containing only $\lceil n/2 \rceil$ points, as in step~\labelcref{alg-step:assign-block} the point $x$ has yet to be placed.
  \end{claimproof}

  The following two claims show that the algorithm is correct.

  \begin{claim*}
    At termination, $G(A) \leq 2\sqrt{n}$.
  \end{claim*}
  \begin{claimproof}
    Since every point is placed in the left-most empty cell of a block, there is at most one gap per block, so $G(A) \leq N_{2} \leq 2\sqrt n$.
  \end{claimproof}

  \begin{claim*}
    At termination, $c(A) \leq 11\sqrt{n}\OPT(X)$.
  \end{claim*}
  \begin{claimproof}
    The cost between two neighboring blocks is at most $\OPT(X)$, so the total cost between neighboring blocks is at most $N_2\OPT(X)$.
    It remains to bound the cost between neighboring cells inside a block.
    So let $x$ and $x'$ be points placed in neighboring cells inside a block $B$.

    Consider first the case where the net was not reset (steps~\labelcref{alg-step:unassign-all} to~\labelcref{alg-step:empty-net}) between $x$ and $x'$ being placed.
    Then $B$ was assigned to the same center $c$ with the same radius $r$ both when $x$ and $x'$ were placed.
    Let $X'$ be the set of points read right after both $x$ and $x'$ were placed.
    By the triangle inequality, $\OPT(X') \leq \OPT(X)$, so by~\cref{2-approx}, $d(x, x') \leq d(x, c) + d(c, x') \leq 2r = 8\MST(X')/N_{1} \leq 8\OPT(X')/N_{1} \leq 8\OPT(X)/N_{1}$.
    Since at most $\lceil n/2 \rceil$ points are placed, the total cost between such pairs of points is at most $(\lceil n/2 \rceil - 1)8\OPT(X)/N_{1} \leq 4n\OPT(X)/N_{1}$.

    Consider now the case where the net was reset between $x$ and $x'$ being placed.
    Then $d(x, x') \leq \MST(X')$ where $X'$ is the set of points read at any point after placing both $x$ and $x'$.
    This is a worse bound, but this case can only appear once per block per reset of the net.
    So the total cost for such pairs of points is at most $N_{2}\MST(X')$ for each reset, where $X'$ is the set of points read when resetting.
    Let $X_{1}$ be the set of points read at the point of a reset, or at the beginning of the algorithm, and $X_{2}$ be the set of points read at the point of the following reset or when the algorithm terminates.
    Right before the latter reset, $|C| > N_{1}$ and $r = 4\MST(X_{1})/N_{1}$.
    This is true even if $X_{1}$ is the set at the beginning of the algorithm, as then $X_{1}=\emptyset$ and $r=0$.
    So~\Cref{net-size} gives us that $2\MST(X_{2}) \geq (|C|-1)r \geq N_{1}4\MST(X_{1})/N_{1} = 4\MST(X_{1})$, which simplifies to $\MST(X_{2}) \geq 2\MST(X_{1})$.
    The total cost of this type across all rebuilds is thus no more than the geometric series $N_{2}\sum_{i=0}2^{-i}\MST(X) = 2N_{2}\MST(X) \leq 2N_{2}\OPT(X)$, where the last inequality uses~\Cref{2-approx}.

    The total cost $c(A)$ at termination is thus at most $(N_{2} + 4n/N_{1} + 2N_{2})\OPT(X)$.
    It can be checked that $N_{2}+4n/N_{1}+2N_{2} \leq 11\sqrt{n}$, finishing the proof.
  \end{claimproof}

  The combination of the above claims finishes the proof.
\end{proof}

We are now ready to present our full algorithm, which recursively applies Fill-Most-Blocks, analogously to~\cite{online-packing,online-tsp}.
The algorithm begins by placing the first half of the input points into the array $A$ using Fill-Most-Blocks.
It then treats the remaining empty cells of $A$ as a contiguous array $A_{\mathrm{empty}}$.
The algorithm proceeds recursively on $A_{\mathrm{empty}}$, treating it as a standard array; however, when a point is placed in $A_{\mathrm{empty}}$, it is actually placed into the corresponding cell of $A$.
The exact algorithm is as follows.

\begin{tcolorbox}[beforeafter skip=10pt]
  \textbf{Recursively-Fill-Most-Blocks($n, A, X$)}
  \begin{enumerate}
    \item If $n = 0$ then return.
    \item Let $X_{\mathrm{prefix}}$ be the stream consisting the first $\lceil n/2 \rceil$ points of $X$.
          \item\label{alg-step:fill-half} Fill-Most-Blocks$(n, A, X_{\mathrm{prefix}})$.
    \item Consider the empty cells of $A$ as one contiguous array $A_{\mathrm{empty}}$.
    \item Let $X_{\mathrm{suffix}}$ be the stream consisting the remaining $\lfloor n/2 \rfloor$ points of $X$.
    \item Recursively-Fill-Most-Blocks$(\lfloor n/2 \rfloor, A_{\mathrm{empty}}, X_{\mathrm{suffix}})$.
  \end{enumerate}
\end{tcolorbox}

\begin{theorem}\label{full-alg}
  Let $(M, d)$ be a metric space and $X$ be an online stream of $n$ points in $M$.
  Let $A$ be an empty array of length $n$.
  The deterministic algorithm \textnormal{Recursively-Fill-Most-Blocks$(n, A, X)$} irrevocably places each point from $X$ in an empty cell of $A$, such that when all points have been placed, we have $c(A) \leq 52\sqrt n\OPT(X)$.
\end{theorem}
\begin{proof}
  We will show a cost of at most $15(2+\sqrt 2)\sqrt n\OPT(X)$ using strong induction on $n$.
  The base case $n=0$ is trivial, so let us handle the inductive step.
  Let $A'$ denote the array $A$ after step~\labelcref{alg-step:fill-half}.
  Then $G(A') \leq 2\sqrt n$ and $c(A') \leq 11\sqrt n\OPT(X_{\mathrm{prefix}}) \leq 11\sqrt n\OPT(X)$ by~\Cref{half-alg} and the triangle inequality.
  By induction, we have $c(A_{\mathrm{empty}}) \leq 15(2+\sqrt 2)\sqrt{n/2}\OPT(X)$.
  The final cost $c(A)$ is the sum of $c(A')$, $c(A_{\mathrm{empty}})$, and the costs between neighboring cells in $A$ where exactly one of the cells was empty in $A'$.
  The latter is at most $2G(A')\OPT(X) \leq 4\sqrt n\OPT(X)$, as there are at most two such pairs of neighboring cells per gap in $A'$.
  The total cost becomes $c(A) \leq (11+15(2+\sqrt 2)/\sqrt 2 + 4)\sqrt n\OPT(X) = 15(2+\sqrt 2)\OPT(X)$.
\end{proof}

From~\Cref{full-alg}, we immediately get~\Cref{sqrt-n-alg,sqrt-n-alg-euclidean}.
We have thus generalized the optimal upper bound for online sorting~\cite{online-tsp} to online metric TSP, in particular improving the best-known upper bound for online TSP in $d$-dimensional Euclidean space.
In~\Cref{sec:optimality} we show, that our algorithm is optimal for both general online metric TSP, online TSP in $d$-dimensional Euclidean space, and more.

\section{Optimality}%
\label{sec:optimality}

It follows directly from a known lower bound for online sorting~\cite{online-packing,online-tsp} that our algorithm is optimal.
Recall that online sorting is equivalent to online TSP in 1-dimensional Euclidean space.

\begin{theorem}[Theorem 1 in~\cite{online-tsp}]\label{known-lb}
  The (deterministic and randomized) competitive ratio of online TSP in the Euclidean unit interval $[0, 1]$ is $\Omega(\sqrt n)$.
\end{theorem}

In particular, this gives a lower bound of $\Omega(\sqrt n)$ for general online metric TSP, showing that our $O(\sqrt n)$ algorithm Recursively-Fill-Most-Blocks from~\Cref{sec:our-algorithm} is optimal.
Perhaps surprisingly, online metric TSP is no harder than online sorting.

Since the Euclidean unit interval lies inside $d$-dimensional Euclidean space,~\Cref{known-lb} also shows that our $O(\sqrt n)$ algorithm is optimal for online TSP in $d$-dimensional Euclidean space.
We have thus closed the gap for this problem, improving upon the best known upper bound of $O(2^{d}\sqrt{dn \log n})$~\cite{online-tsp}, notably removing the dependence on dimension.

Studying the proof of the above lower bound, we find the following generalization.
Intuitively, the generalization gives us a lower bound of $\Omega(\sqrt n)$ for online TSP in any metric space where we can draw a straight line segment (of length $\ell$ with endpoints $a_{0}$ and $a_{1}$) and pick $m$ evenly spaced points along it.

\begin{corollary}\label{generalized-lb}
  Let $(M, d)$ be a metric space, such that for every $m \in \N$, there exist two points $a_{0}, a_{1} \in M$, and a set $X \subseteq M$ of $m$ points, such that, for $\ell = \OPT(X \cup \{a_{0}, a_{1}\})$,
  \begin{itemize}
    \item $d(x,y) \geq \ell/m$ for all distinct $x, y \in X$, and
    \item $d(a_{0}, x) + d(x, a_{1}) \geq \ell$ for all $x \in X$.
  \end{itemize}
  Then the competitive ratio of online TSP in $(M, d)$ is $\Omega(\sqrt n)$.
\end{corollary}
\begin{proof}
  This follows from the proof of~\Cref{known-lb} presented in Section 2 of~\cite{online-tsp}.
  They work in the Euclidean unit interval, but only use the points $a_{0}=0$, $a_{1}=1$, and the set of points $X = \{0, 1/\sqrt n, \ldots, (\sqrt n-1)/\sqrt n\}$.
  Ignoring the concrete values of the points $a_{0}$, $a_{1}$, and those in $X$, it can easily be checked, that they only use the properties stated above, where $m = \sqrt n$ and $\ell = 1$.
  The proof also follows through for any other $\ell > 0$, where all distances in their proof simply scale by $\ell$.
\end{proof}

This tells us that our algorithm is also optimal for online TSP in e.g.\ (subgroups of) normed vector spaces and Riemannian manifolds.

\subsection{Impossibility of optimality for every fixed metric space}

We have seen, that our $O(\sqrt n)$ algorithm is optimal for general online metric TSP, as well as for online TSP in many fixed metric spaces.
It is not optimal for every fixed metric space, though, as an algorithm with competitive ratio $O(\log n)$ is known for the uniform metric~\cite{online-tsp}.
Motivated by this gap, one might ask whether there exists a stronger algorithm, which optimally solves online TSP in $(M, d)$, for every fixed metric space $(M, d)$.
In this subsection, we show that no such algorithm exists.
This hints, that the weaker optimality of our algorithm is the best one can hope for.

Specifically, we show that no algorithm is optimal both for the uniform metric and for a general metric space.
This is the content of~\Cref{no-log-n-and-sqrt-n}, which we restate and prove below.

\impossibilityResult*
\begin{proof}
  Assume for the sake of contradiction that such an algorithm $\A$ exists.
  Let us first informally explain our basic idea.
  Let $U$ be a set of $n^{4/5}$ points with pairwise distance 1, and let $x$ be a point with distance $n^{4/5}$ to every point in $U$.
  Consider an input consisting $n^{1/5}$ consecutive copies of $U$.
  Then the optimal offline cost is $\OPT(U) = |U|-1 = n^{4/5}-1$.
  The only way for $\A$ to match this optimal cost, would be to maintain $\Omega(|U|)=\Omega(n^{4/5})$ gaps until completion.
  Similarly, since $\A$ actually must obtain nearly-optimal cost $O(n^{4/5}\log n)$, it must have at least $n^{3/5}$ gaps at some point.
  But when this happens, we can trick $\A$ by changing the remaining input to be copies of $x$, filling up all these gaps.
  Then the cost becomes at least $n^{3/5}n^{4/5}$ which breaks the promise of competitive ratio $O(\sqrt n)$, since $\OPT(U \cup \{x\}) = 2n^{4/5}-1$.

  We now formally prove the \lcnamecref{no-log-n-and-sqrt-n}.
  Let $\X$ be the random input served by Oblivious-Random-Adversary$(n)$, which we define below.

  \begin{tcolorbox}[beforeafter skip=10pt]
    \textbf{Oblivious-Random-Adversary($n$)}
    \begin{enumerate}
      \item Let $U$ be a set of $n^{4/5}$ points with pairwise distance 1.
      \item Let $x$ be a point with distance $n^{4/5}$ to every point in $U$.
      \item While less than $n$ points have been served:
            \begin{enumerate}
              \item With probability $n^{-3/5}$:
                    \begin{enumerate}
                      \item Let $m$ be the number of points served.
                      \item Serve $n-m$ copies of $x$.
                    \end{enumerate}
              \item Otherwise:
                    \begin{enumerate}
                      \item\label{epoch} Serve a copy of the points in $U$.
                    \end{enumerate}
            \end{enumerate}
    \end{enumerate}
  \end{tcolorbox}

  By closely following the proof of~\Cref{known-lb}, replacing their random input by $\X$, we get that the expected cost of any deterministic algorithm on $\X$ is $\Omega(n)$.
  We provide a full proof in~\Cref{app:lb-details}, restating this claim as~\Cref{impossibility-lb}.

  Let $\A(\X)$ denote the cost of $\A$ on $\X$.
  Considering $\A$ as a random variable over deterministic algorithms, we get $\E[\A(\X)] \in \Omega(n)$.
  We will now derive an upper bound contradicting this lower bound.
  To utilize our assumed competitive ratio for the uniform metric, note that $\X$ follows the uniform metric when $x \not\in \X$.
  From the definition of Oblivious-Random-Adversary, we have $\Pr[x \in \X] = 1 - \Pr[x \not\in \X] = 1 - (1-n^{-3/5})^{n^{1/5}} \leq n^{-2/5}$, using Bernoulli's inequality.
  Using this bound, we get
  \begin{align*}
    \E[\A(\X)]
     & = \Pr[x \in \X]\E[\A(\X) \mid x \in \X] + \Pr[x \not\in \X]\E[\A(\X) \mid x \not\in \X] \\
     & \leq n^{-2/5}\E[\A(\X) \mid x \in \X] + \E[\A(\X) \mid x \not\in \X].
  \end{align*}
  By assumption, $\A$ has competitive ratio $O(\sqrt n)$ in general, and competitive ratio $O(\log n)$ when $x$ is not in the input.
  It's easy to see that $\OPT(\X) \in O(n^{4/5})$, so we have $\E[\A(\X) \mid x \in \X] \in O(n^{4/5}\sqrt n)$ and $\E[\A(\X) \mid x \not\in \X] \in O(n^{4/5}\log n)$.
  In total, $\E[\A(\X)] \in O(n^{-2/5}n^{4/5}\sqrt{n} + n^{4/5}\log n) \subseteq o(n)$, contradicting $\E[\A(\X)] \in \Omega(n)$.
\end{proof}

\section{Open questions}%
\label{sec:open-questions}

Both~\cite{online-packing} and~\cite{online-tsp} additionally study a variant of the problem with extra space.
In this variant, the array is of length $\gamma n$ for some $\gamma > 1$, and empty cells are ignored in the cost function.
A significant gap between the best-known upper and lower bound remains, even for online sorting of reals with extra space~\cite{online-packing}.
For the uniform metric, a tight $\Theta(1+\log(\gamma/(\gamma-1)))$ competitive ratio is known~\cite{online-tsp}.
We repeat it as an open question to tighten the gap for online sorting of reals with extra space, and suggest the introduction of an algorithm for online metric TSP with extra space.

\section{Acknowledgments}
I am grateful to Anders Aamand, Mikkel Abrahamsen, Théo Fabris, Jonas Klausen, and Mikkel Thorup for helpful and inspiring conversations.
I would also like to thank the anonymous ESA 2025 reviewers for their valuable comments.

\bibliography{online-metric-tsp.bib}

\newpage
\appendix
\section{Deferred proofs}

\subsection{Deferred bound on proxy}%
\label{app:proxy-details}

We restate and prove~\Cref{2-approx}.

\proxy*
\begin{proof}
  Recall that $\OPT(X)$ is the length of some (shortest) walk covering all points in $X$, so in particular it is the total weight of some connected multigraph on $X$.
  Since $\MST(X)$ is the minimum total weight of a connected graph on $X$, we have $\MST(X) \leq \OPT(X)$.
  To show the remaining inequality, consider walking along the face created by placing a minimum spanning tree in the plane.
  This gives a walk of length $2\MST(X)$ covering $X$.
\end{proof}

\subsection{Deferred lower bound}%
\label{app:lb-details}

We prove the following \namecref{impossibility-lb}, which is applied in the proof of~\Cref{no-log-n-and-sqrt-n}.

\begin{lemma}\label{impossibility-lb}
  Let $\A$ be a deterministic algorithm for online metric TSP.\@
  Let $\X$ be the random input served by Oblivious-Random-Adversary$(n)$.
  Then $\E[\A(\X)] \in \Omega(n)$.
\end{lemma}
\begin{proof}
  Our proof follows the proof of Theorem 1 in~\cite{online-tsp} (restated as~\Cref{known-lb} in this paper) almost verbatim, with slight adjustments to accommodate our context.
  We replace their online sorting problem by online metric TSP, and replace their random input by $\X$.
  To accommodate this, we modify many of the quantities in their proof, but all the same arguments still work.

  Let us introduce the necessary notation and terminology.
  For a partially filled size-$n$ array $T$, let $|T|$ denote the number of non-empty cells in $T$.
  Let $\A(T)$ denote the final array produced by $\A$ when the remaining $n-|T|$ input points are served by Oblivious-Random-Adversary$(n)$.
  Define $f(T) = c(\A(T))-c(T)$, representing the cost of filling $T$.
  Then $f(T) = 0$ if $|T|=n$, and our goal is to show $\E[f(\emptyset)] \in \Omega(n)$, using $\emptyset$ for the empty size-$n$ array.

  We say that an array $T_{2}$ can be \emph{obtained from} an array $T_{1}$ if $T_{2}[i] = T_{1}[i]$ for all indices where $T_{1}[i]$ is non-empty.
  For such a pair of arrays, we write $c(T_{1}, T_{2}) = c(T_{2})-c(T_{1})$, representing the cost of transforming $T_{1}$ into $T_{2}$.

  In the input sequence, we refer to one copy of $U$ as an \emph{epoch} (as served in step~\labelcref{epoch}).
  Let $\T(T)$ be the set of mappings that can be obtained by inserting an epoch into $T$.
  That is, $\T(T)$ contains all mappings $T'$ which can be obtained from $T$ where the difference between $T'$ and $T$ corresponds to the elements of an epoch.
  Note that $|T'| = |T|+\sqrt n$ for $T' \in \T(T)$.

  Let $\Low = \{T \mid G(T) \leq n^{4/5}/8\}$ and $\High = \{T \mid G(T) > n^{4/5}/8\}$ be the sets of all partially filled arrays with a low/high number of gaps.

  \begin{claim*}[Adapted from Lemma 9 in~\cite{online-tsp}]
    Let $T_{1}, T_{2} \in \Low$ with $T_{2} \in \T(T_{1})$.
    Then $c(T_{1},T_{2}) \geq 3n^{4/5}/16$.
  \end{claim*}
  \begin{claimproof}
    If an element is placed between two empty cells, the number of gaps will increase by one.
    If it is placed next to one or two occupied cells, the number of gaps stays constant or is reduced by one---call such an insertion an attachment.
    As $G(T_{2}) \leq G(T_{1})+n^{4/5}/8$, at least $7n^{4/5}/16$ of the $n^{4/5}$ insertions of the epoch must be attachments.

    An attachment incurs cost at least 1 unless the point(s) in the neighboring cell(s) and the newly inserted point are identical.
    As all points withing an epoch are distinct, an attachment can only be without cost if the neighboring cell was occupied in $T_{1}$ (that is, before the epoch).
    As $G(T_{1}) \leq n^{4/5}/8$, at most $n^{4/5}/4$ occupied cells border an empty cell.
    The remaining $n^{4/5}(7/16 - 1/4)=3n^{4/5}/16$ attachments will incur non-zero cost.
  \end{claimproof}

  As a shorthand, let $(T+x)$ be the arrays obtained by filling all gaps in $T$ with $x$s, where $x$ is the point defined in Oblivious-Random-Adversary.

  \begin{claim*}[Adapted from Lemma 10 in~\cite{online-tsp}]
    Let $T$ be a partially filled array not containing $x$.
    Then $c(T, (T+x)) \geq G(T)n^{4/5}$.
  \end{claim*}
  \begin{claimproof}
    Each gap in $T$ is bordered by at least one non-empty cell $A[i]$.
    Since $A[i] \in U$, we have $d(A[i], x) = n^{4/5}$.
  \end{claimproof}

  For a partially filled array $T$, the expected cost of filling $T$ can be bounded from below by
  \[
    \E[f(T)] \geq n^{-3/5}c(T, (T+x)) + (1-n^{-3/5})\min_{T'\in\T(T)}(c(T,T')+\E[f(T')]),
  \]
  and by the second claim, we thus have
  \begin{equation}\label{fill-lb}
    \E[f(T)] \geq G(T)n^{1/5} + (1-n^{-3/5})\min_{T'\in\T(T)}(c(T,T')+\E[f(T')]).
  \end{equation}

  Let $\Low(i) = \min\{\E[f(T)] \mid T \in \Low, |T| = n-in^{4/5}\}$ and $\High(i) = \min\{\E[f(T)] \mid T \in \High, |T| = n-in^{4/5}\}$ be the minimum expected cost of filling any array with $in^{4/5}$ empty cells, and which contains a low/high number of gaps, for $i \in \{0,\ldots,n^{1/5}\}$, respectively $i \in \{0,\ldots,n^{1/5}-1\}$.
  (Note that $\High(n^{1/5})$ is undefined as an empty array cannot have a high number of gaps.)

  Combining~\eqref{fill-lb} with the first claim, we obtain
  \begin{align*}
    \Low(i) &\geq (1-n^{-3/5})\min\{3n^{4/5}/16 + \Low(i-1), \High(i-1)\},\\
    \High(i) &\geq n^{4/5}/8 + (1-n^{-3/5})\min\{\Low(i-1), \High(i-1)\},
  \end{align*}
  where $\Low(0)=\High(0)=0$.
  Our next claim, leads to the lower bound.

  \begin{claim*}[Adapted from Lemma 11 in~\cite{online-tsp}]
    $\Low(n^{1/5}) \in \Omega(n)$.
  \end{claim*}
  \begin{claimproof}
    We proceed by strong induction on $i$.
    For all $i' \in \{0, \ldots, i-1\}$ assume
    \begin{align*}
      \High(i') &\geq i'n^{4/5}(1-n^{-3/5})^{i'}/16,\\
      \Low(i') &\geq (i'-1)n^{4/5}(1-n^{-3/5})^{i'}/16.
    \end{align*}
    Then
    \begin{align*}
      \High(i)
      &\geq n^{4/5}/8 + (1-n^{-3/5})\min\{\Low(i-1), \High(i-1)\}\\
      &\geq n^{4/5}/8 + (i-2)n^{4/5}(1-n^{-3/5})^{i}/16\\
      &\geq in^{4/5}(1-n^{-3/5})^{i}/16,
    \end{align*}
    and
    \begin{align*}
      \Low(i)
      &\geq (1-n^{-3/5})\min\{3n^{4/5}/16 + \Low(i-1), \High(i-1)\}\\
      &\geq (1-n^{-3/5})\min\{3 + (i-2)(1-n^{-3/5})^{i-1}, (i-1)(1-n^{-3/5})^{i-1}\}n^{4/5}/16\\
      &\geq (i-1)n^{4/5}(1-n^{-3/5})^{i}/16.
    \end{align*}
    Finally, note that $(1-n^{-3/5})^{n^{1/5}} \geq 0.3$ for $n \geq 3$, proving the claim.
  \end{claimproof}

  As the empty array $\emptyset$ is contained in $\Low$, we have $\E[f(\emptyset)] \geq \Low(n^{1/5}) \in \Omega(n)$.
\end{proof}

In~\cite{online-tsp}, they finish their proof by applying Yao's minimax principle~\cite[Proposition 2.6]{ra}.
In our case, this would give us a lower bound of $\Omega(n^{1/5})$ on the competitive ratio for online metric TSP, since $\OPT(\X) \in O(n^{4/5})$.
While this is not directly interesting to us,~\Cref{impossibility-lb} helps us obtain a contradiction in the proof of~\Cref{no-log-n-and-sqrt-n}.

\end{document}